\documentclass{llncs}

\title{Improved ESP-index: a practical self-index for highly repetitive texts \thanks{This work was supported by JSPS KAKENHI(24700140,23680016) and the JST PRESTO program.}}

\author{Yoshimasa Takabatake\inst{1} 
  \and Yasuo Tabei\inst{2}
  \and Hiroshi Sakamoto\inst{1}}
\institute{
  Kyushu Institute of Technology~\email{\{takabatake,hiroshi\}@donald.ai.kyutech.ac.jp}
  \and 
  PRESTO, Japan Science and Technology Agency~\email{tabei.y.aa@m.titech.ac.jp}
}

\usepackage{makeidx}
\usepackage{algorithm}
\usepackage{algorithmicx}
\usepackage{algpseudocode}
\usepackage{amsmath}
\usepackage{url}
\usepackage[dvips]{graphicx}
\usepackage{comment}

\newcommand{\logstar}{\lg^*\hspace{-.9mm}}

\begin{document}

\maketitle

\begin{abstract}
While several self-indexes for highly repetitive texts exist, developing a practical self-index applicable to real world repetitive texts remains a challenge. 
ESP-index is a grammar-based self-index on the notion of edit-sensitive parsing (ESP), 
an efficient parsing algorithm that guarantees upper bounds of parsing discrepancies between different appearances of the same subtexts in a text.
Although ESP-index performs efficient top-down searches of query texts, it has a serious issue on binary searches for finding appearances of variables for a query text, 
which resulted in slowing down the query searches. 
We present an improved ESP-index (ESP-index-I) by leveraging the idea behind succinct data structures for large alphabets. 
While ESP-index-I keeps the same types of efficiencies as ESP-index about the top-down searches, it avoid the binary searches using fast rank/select operations.
We experimentally test ESP-index-I on the ability to search query texts and extract subtexts from real world repetitive texts on a large-scale, and 
we show that ESP-index-I performs better that other possible approaches. 
\end{abstract}

\section{Introduction}
Recently, highly repetitive text collections have become common. 
Examples are human genomes, version controlled documents and source codes in repositories. 
In particular, the current sequencing technology enables us to sequence individual genomes in a short time,
resulting in generating a large amount of genomes, perhaps millions of genomes in the near future. 
There is therefore a strong demand for developing powerful methods to store and process such repetitive texts on a large-scale. 

Grammar compression is effective for compressing and processing repetitive texts, and 
it builds a {\em context free grammar}~(CFG) that generates a single text. 
There are two types of problems: (i) building as small as possible of a CFG generating an input text and
(ii) representing the obtained CFG as compactly as possible for various applications. 
Several methods  have been presented for type~(i).
Representative methods are RePair~\cite{Larsson2000} and LCA~\cite{Sakamoto09}. 
Methods for type~(ii) have also been presented for processing repetitive texts, e.g., pattern matching~\cite{Yamamoto2011}, pattern mining~\cite{Goto13} and edit distance computation~\cite{Hermelin09}. 

Self-indexes aim at representing a collection of texts in a compressed format that supports extracting subtexts of arbitrary positions and also provides query searches on the collection, 
and are fundamental in modern information retrieval.
However, developing a grammar-based self-index remains a challenge, 
since a grammar-compressed text forms a tree structure named parse tree and 
variables attached to its nodes do not necessarily encode all portions of a text, which makes the problem of searching query texts from grammar compressed texts difficult. 
Claude et al.~\cite{Claude2010,Claude12} presented a grammar-based self-index named {\em SLP-index}.
SLP-index uses two step approaches: (i) it finds variables called first occurrences that encode all prefixes and suffixes of a query text 
by binary searches on compactly encoded CFGs and (ii) then it discovers the remaining occurrences of the variables. 
However, finding first occurrences for moderately long queries is computationally demanding because the method needs to perform binary searches as many times as the query length, which resulted in reducing the practical usage of their method. 

{\em Edit-sensitive parsing} (ESP)~\cite{Cormode07} is an efficient parsing algorithm developed for approximately computing edit distances with moves between texts.
ESP builds from a given text a parse tree that guarantees upper bounds of parsing discrepancies between different appearances of the same subtext. 
Maruyama et al.~\cite{ESP} presented another grammar-based self index called {\em ESP-index} on the notion of ESP. 
ESP-index represents a parse tree as a {\em directed acyclic graph}~(DAG) and then encodes the DAG into 
succinct data structures for ordered trees and permutations. 
Unlike SLP-index, it performs top-down searches for finding candidates of appearances of a query text 
on the data structure by leveraging the upper bounds of parsing discrepancies in ESP. 
However, it has a serious issue on binary searches for finding appearances of variables.

\begin{table}[bt]
\begin{center}
{\footnotesize
  \caption{Comparison with existing methods. Searching time and extraction time is presented in big $O$ notation that is omitted for space limitations. $u$ is text length, $m$ is the length of a query text, $n$ is the number of variables in a grammar, $\sigma$ is alphabet size, $z$ is the number of pharases in LZ77, $d$ is the length of nesting in LZ77, 
  $occ$ is the number of occurrences of query text in a text, $occ_c$ is the number of candidate apperances of queries, $\lg^*$ is the iterated logarithm and $\epsilon$ is a real value in $(0,1)$. $\lg$ stands for $\log_2$.}
}
\vspace{-0.2cm}  
\label{tbl:comp}
{\footnotesize
\begin{tabular}{l|c|c|c}
                  & Space (bits)  & Searching time & Extraction time \\  
\hline
LZ-index~\cite{Navjda03}   & $z\lg{u}+5z\lg{\sigma}$ & $m^2d+(m+occ)\lg{z}$ & $md$ \\
                  & $-z\lg{z}+o(u)+O(Z)$ &  & \\
\hline
Gagie et al.~\cite{Gagie12} & $2n\lg{n}+O(z\lg{u}$ & $m^2+(m+occ)\lg\lg{u}$ & $m+\lg\lg{u}$ \\
                            & $+z\lg{z}\lg\lg{z})$ &                         &  \\
\hline
SLP-index~\cite{Claude2010,Claude12}& $n\lg{u}+O(n\lg{n})$ & $(m^2+h(m+occ))\lg{n}$ & $(m+h)\lg{n}$ \\
\hline
ESP-index~\cite{ESP} & $n\lg{u}+(1+\epsilon)n\lg{n}$ & $(1/\epsilon)(m\lg{n}$ & $(1/\epsilon)(m+\lg{u})$ \\
                     & $+4n+o(n)$    & $+occ_c\lg{m}\lg{u})\lg^*{u}$         & \\
\hline\hline 
ESP-index-I & $n\lg{u}+n\lg{n}$ & $(\lg{\lg{n}})(m$ & $(\lg\lg{n})(m+\lg{u})$ \\
            & $+2n+o(n\lg{n})$   & $+occ_c\lg{m}\lg{u})\lg^*{u}$        &                    \\
\end{tabular}
}
\end{center}
\end{table}

In this paper, we present an {\em improved ESP-index} (ESP-index-I) for fast query searches.  
Our main contribution is to develop a novel data structure for encoding a parse tree built by ESP. 
Instead of encoding the DAG into two ordered trees using succinct data structures in ESP-index, 
ESP-index-I encodes it into a bit string and an integer array by leveraging the idea behind rank/select dictionaries for large alphabets~\cite{Golynski06}.
Instead of performing binary searches for finding variables on data structures in SLP-index and ESP-index, 
ESP-index-I computes fast select queries in $O(1)$ time, resulting in faster query searches. 
Our results and those of existing algorithms are summarized in Table~\ref{tbl:comp}.

Experiments were performed on retrieving query texts from real-world large-scale texts. 
The performance comparison with other algorithms demonstrates ESP-index-I's superiority.

\section{Preliminaries}
The length of string $S$ is denoted by $|S|$, and the cardinality of a set $C$ is similarly denoted by $|C|$.
The set of all strings over the alphabet $\Sigma$ is denoted by $\Sigma^*$, and let $\Sigma^i = \{w\in\Sigma^*\mid |w|=i\}$.
We assume a recursively enumerable set ${\cal X}$ of variables with $\Sigma \cap {\cal X} = \emptyset$. 
The expression $a^+$ $(a\in \Sigma)$ denotes the set $\{a^k\mid k\geq 1\}$, and
string $a^k$ is called a {\em repetition} if $k\geq 2$.
Strings $x$ and $z$ are said to be a prefix and suffix of $S=xyz$, respectively.
In addition, $x,y,z$ are called substrings of $S$.
$S[i]$ and $S[i,j]$ denote the $i$-th symbol of string $S$ and
the substring from $S[i]$ to $S[j]$, respectively.
$\lg$ stands for $\log_2$.
We let $\lg^{(1)}u = \lg u$, $\lg^{(i+1)}u = \lg \lg^{(i)}u$, and
$\logstar u = \min\{i\mid \lg^{(i)}u\leq 1\}$.
In practice, we can consider $\logstar u$ to be constant, 
since $\logstar u\leq 5$ for $u\leq 2^{65536}$.

\subsection{Grammar compression}
A CFG is a quadruple $G=(\Sigma, V, D, X_s)$ where $V$ is a finite subset of $\cal X$, 
$D$ is a finite subset of $V \times (V \cup \Sigma)^*$ of production rules, 
and $X_s \in V$ represents the start symbol.
Variables in $V$ are called nonterminals. 
We assume a total order over $\Sigma \cup V$. 
The set of strings in $\Sigma^*$ derived from $X_s$ by $G$ is denoted by $L(G)$. 
A CFG $G$ is called {\em admissible} if for any $X \in {\cal X}$ there is exactly one production rule $X\to \gamma \in D$ and $|L(G)| = 1$.
An admissible $G$ deriving a text $S$ is called a grammar compression of $S$. 
The size of $G$ is the total of the lengths of strings on the right hand sides of all production rules; it is denoted by $|G|$.
The problem of grammar compression is formalized as follows:
\begin{definition}[Grammar Compression] \label{def:gc}
Given a string $w \in \Sigma^*$, compute a small, admissible $G$ that derives only $w$.
\end{definition}


$S(D) \in \Sigma^*$ denotes the string derived by $D$ from a string $S\in(\Sigma\cup V)^*$.
For example, when $S=aYY$, $D=\{X\to bc,Y\to Xa\}$ and $\Sigma=\{a,b,c\}$,
we obtain $S(D)=abcabca$.
$|X|$, also denoted by $|X(D)|$, represents the length of the string derived by $D$ from $X\in V$.

We assume any production rule $X\to\gamma$ satisfies $|\gamma| = 2$
because any grammar compression $G$ can be transformed into $G^\prime$
satisfying $|G^\prime|\leq 2|G|$.

The parse tree of $G$ is represented by a rooted ordered binary tree such that
internal nodes are labeled by variables, and the yields, i.e.,
the sequence of labels of leaves is equal to $w$.
In a parse tree, any internal node $Z\in V$ corresponds to the production rule $Z\to XY$, and it has a left child labled by $X$ 
and a right child labeled by $Y$. 
The height of a tree is the length of the longest one among paths from the root to leaves.

\subsection{Phrase and reverse dictionaries}
A phrase dictionary is a data structure for directly accessing a digram $X_iX_j$ from a given $X_k$ if $X_k\to X_iX_j\in D$.
It is typically implemented by an array requiring $2n\log n$ bits for storing $n$ production rules.
In this paper, $D$ also represents its phrase dictionary.
A reverse dictionary $D^{-1}: (\Sigma \cup {\cal X})^2 \to {\cal X}$ is a mapping from a given digram to a nonterminal symbol. 
$D^{-1}$ returns a nonterminal $Z$ associated with a digram $XY$ if $Z\to XY \in D$; otherwise, it creates a new nonterminal symbol $Z' \notin V$ and returns $Z'$. 
For example, if we have a phrase $D=\{X_1\to ab, X_2\to cd\}$, then $D^{-1}(a,b)$ returns $X_1$, while $D^{-1}(b,c)$ creates a new nonterminal $X_3$ and returns it. 

\subsection{Rank/select dictionaries}
Our method represents CFGs using a rank/select dictionary, a succinct data structure for a bit string $B$~\cite{Jacobson89} 
supporting the following queries: $\mbox{rank}_c(B,i)$ returns the number of occurrences of $c \in \{0,1\}$ in 
$B[0,i]$; $\mbox{select}_c(B,i)$ returns the position of the $i$-th occurrence of $c\in\{0,1\}$ in $B$; 
$\mbox{access}(B,i)$ returns $i$-th bit in $B$.
Data structures with only the $|B| + o(|B|)$ bit storage to achieve $O(1)$ time rank and select queries~\cite{Raman07} have been presented. 

GMR~\cite{Golynski06} is a rank/select dictionary for large alphabets and supports 
rank/ select/access queries for general alphabet strings $S \in \Sigma^*$. 
GMR uses $n\log{n}+o(n\log{n})$ bits while computing both rank and access queries in $O(\log{\log{|\Sigma|}})$ times and 
also computing select queries in $O(1)$ time. 
Space-efficient implementations of GMR are also presented in \cite{Barbay2013}.

\section{ESP-index}

\subsection{Edit-sensitive parsing (ESP)}
In this section, we review a grammar compression based on ESP~\cite{Cormode07},
which is referred to as {\em GC-ESP}. 
The basic idea of GC-ESP is to (i) start from an input string $S\in\Sigma^*$, 
(ii) replace as many as possible of the same digrams in common substrings by the same variables, 
and (iii) iterate this process in a bottom-up manner until $S$ is transformed to a single variable.

In each iteration, 
GC-ESP uniquely divides $S$ into maximal non-overlapping substrings such that 
$S=S_1S_2\cdots S_\ell$ and each $S_i$ is categorized into one of three types: 
(1) a repetition of a symbol;
(2) a substring not including a type1 substring and of length at least $\lg^* |S|$;
(3) a substring being neither type1 nor type2 substrings. 

At one iteration of parsing $S_i$, GC-ESP builds two kinds of subtrees from strings $XY$ and $XYZ$ of length two and three, respectively.
The first type is a $2$-tree corresponding to a production rule in the form of $A\to XY$. 
The second type is a $2$-$2$-tree corresponding to production rules in the forms of $A\to XB$ and $B\to YZ$. 

GC-ESP parses $S_i$ according to its type. 
In case $S_i$ is a type1 or type3 substring, GC-ESP performs the typical left aligned parsing 
where $2$-trees are built from left to right in $S_i$
and a 2-2-tree is built for the last three symbols if $|S_i|$ is odd, as follows:
\begin{itemize} 
  \item If $|S_i|$ is even, GC-ESP builds $A\to S_i[2j-1,2j]$, $j=1,...,|S_i|/2$,
  \item Otherwise, it builds $A\to S_i[2j-1,2j]$ for $j=1,...,(\lfloor |S_i|/2 \rfloor-1)$, 
and builds $A\to B S_i[2j+1]$ and $B\to S_i[2j-1,2j]$ for $j=\lfloor|S_i|/2\rfloor$.
\end{itemize}
In case $S_i$ is a type2 substring, GC-ESP further partitions $S_i$ into several substrings such that 
$S_i=s_1s_2...s_\ell$ ($2\leq |s_j|\leq 3$) using {\em alphabet reduction}~\cite{Cormode07}, which is detailed below.
GC-ESP builds $A\to s_j$ if $|s_j|=2$ or builds $A\to s_j[2,3]$, $B\to s_j[1]A$ otherwise for $j=1,...,\ell$.

GC-ESP transforms $S_i$ to $S'_i$ and parses the concatenated string $S'_i$ $(i=1,\ldots,\ell)$ at the next level of a parse tree (Figure~\ref{fig:esp}). 
In addition, GC-ESP gradually builds a phrase dictionary $D^k$ at $k$th level of a parse tree.
The final dictionary $D$ is the union of dictionaries built at each level of a parse tree, i.e., $D=D^1 \cup D^2 \cup ... \cup D^h$.

\begin{figure*}[t]
\begin{center}
\includegraphics[width=0.9\textwidth]{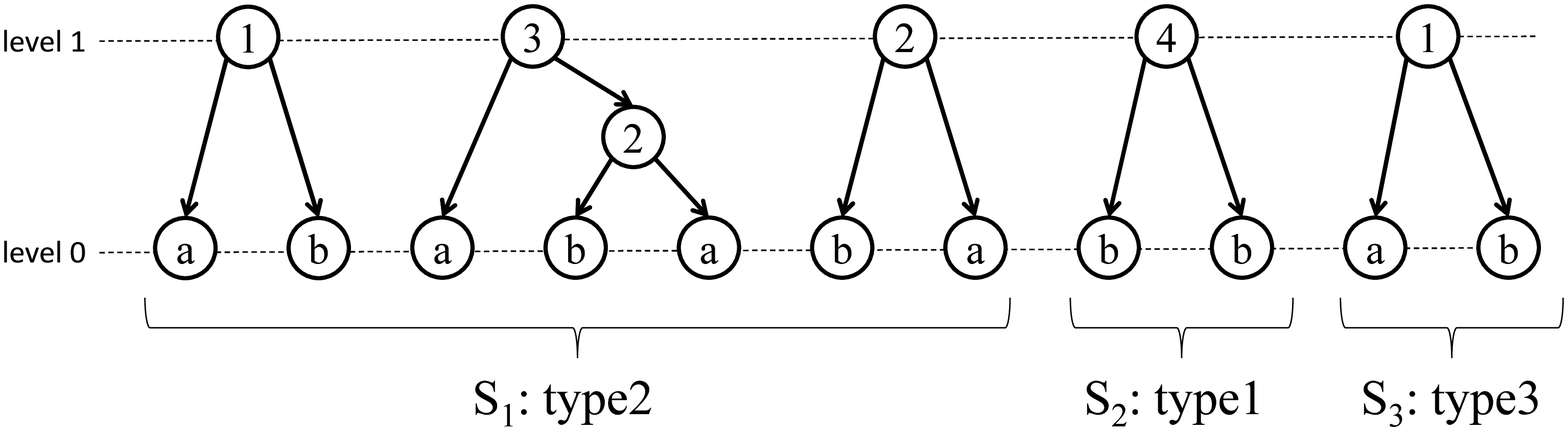}
\end{center}
\vspace{-0.7cm}
\caption{An example for parsing $S=ababababbab$ by GC-ESP.
There are three non-overlapping substrings $S_1=ababa$ of type2, $S_2=bb$ of type1, and $S_3=ab$ of type3.
They are individually parsed with a common reverse dictionary.
The resulted string is $13241$ that is parsed at one higher level.
}
\label{fig:esp}
\end{figure*}

{\bf Alphabet reduction:}
Alphabet reduction is a procedure for partitioning a string into substrings of length 2 and 3. 
Given a type2 substring $S$, consider $S[i]$ and $S[i-1]$ represented as binary integers.
Let $p$ be the position of the least significant bit in which $S[i]$ differs from $S[i-1]$, 
and let ${\it bit}(p,S[i])\in\{0,1\}$ be the value of $S[i]$ at the $p$-th position, where $p$ starts at $0$.
Then, $L[i]=2p+{\it bit}(p,S[i])$ is defined for any $i\geq 2$.
Since $S$ does not contain any repetitions as type2, the resulted string $L=L[2]L[3]\ldots L[|S|]$ does not also contain repetitions, i.e., $L$ is type2.
We note that if the number of different symbols in $S$ is $n$ which is denoted by $[S]=n$, clearly $[L]\leq 2\lg n$.
Setting $S:=L$, the next label string $L$ is iteratively computed until $[L]\leq \logstar |S|$.
At the final $L^*$, $S[i]$ of the original $S$ is called {\em landmark} if $L^*[i]>\max\{L^*[i-1],L^*[i+1]\}$.

After deciding all landmarks, if $S[i]$ is a landmark, we replace $S[i-1,i]$ by a variable $X$
and update the current dictionary with $X\to S[i-1,i]$.
After replacing all landmarks, the remaining maximal substrings 
are replaced by the left aligned parsing.

Because $L^*$ is type2 and $[L^*]\leq \logstar |S|$, 
any substring of $S$ longer than $2\logstar |S|$ must contain at least one landmark.
Thus, we have the following characteristic.

\begin{lemma}\label{lem1:cm}\rm(Cormode and Muthukrishnan~\cite{Cormode07})
Determining the closest landmark to $S[i]$ depends on only $\logstar |S|+5$
contiguous symbols to the left and 5 to the right.
\end{lemma}

This lemma tells us the following.
Let $S$ be type2 string containing $\alpha$ as $S=x\alpha y\alpha z$.
Using Lemma~\ref{lem1:cm}, when $\alpha$ is sufficiently long (e.g., $|\alpha|\geq 2\logstar |S|$),
there is a partition $\alpha=\alpha_1\alpha_2$ such that
$|\alpha_1|=O(\logstar |S|)$ and whether $\alpha_2[i]$ is landmark or not is coincident
in both occurrences of $\alpha$.

Thus, we can construct a consistent parsing for 
all occurrences of $\alpha_2$ in $S$, which almost covers whole $\alpha$ except a short prefix $\alpha_1$.
Such consistent parsing can be iteratively constructed for $\alpha_2$ as the next $S$ while it is sufficiently long.

\begin{lemma}\label{lem2:cm}\rm(Cormode and Muthukrishnan~\cite{Cormode07})
GC-ESP builds from a string $S$ a parse tree of height $h=O(\lg |S|)$ in $O(|S|\logstar |S|)$ time.
\end{lemma}

\subsection{Algorithms}
We present an algorithm for finding all the occurrences of pattern $P$ in $S \in \Sigma^*$ parsed by ESP.
Let $T_S$ be the parsing tree for $S$ by ESP and $D$ be the resulted dictionary for $T_S$.
We consider this problem of embedding a parsing tree $T_P$ of $P$ into $T_S$ as follows.

First, we construct $T_P$ preserving the labeling in $D$ and a new production rule is generated if its phrase is undefined.

Second, $T_P$ is divided into a sequence of maximal {\em adjacent} subtrees rooted by nodes $v_1,\ldots,v_k$ such that
$yield(v_1\cdots v_k)=P$, where $yield(v)$ denotes the string represented by the leaves of $v$ and 
$yield(v_1\cdots v_k)$ denotes the concatenation of strings $yield(v_1),yield(v_2),...,yield(v_k)$.

If $z$ is the lowest common ancestor of $v_1$ and $v_k$, which is denoted by $z=lca(v_1,v_k)$,
the sequence $v_1,\ldots, v_k$ is said to be embedded into $z$, denoted by $(v_1\cdots v_k)\prec z$.
When $yield(v_1\cdots v_k)=P$, $z$ is called an {\em occurrence node} of $P$.

\begin{definition}\label{evidence}
An {\em evidence} of $P$ is defined as a string $Q\in (\Sigma\cup V)^*$ of length $k$ satisfying the following condition: 
There is an occurrence node $z$ of $P$ iff
there is a sequence $v_1\cdots v_k$ such that $(v_1\cdots v_k)\prec z$, $yield(v_1\cdots v_k)=P$, and $L(v_1\cdots v_k)=Q$
where $L(v)$ is the variable of $v$ and $L(v_1\cdots v_k)$ is the concatenation.
\end{definition}
 
An evidence $Q$ transforms the problem of finding an occurrence of $P$ into that of embedding a shorter string $Q$ into $T_S$, 
Since a trivial $Q$ with $Q=P$ always exists, this notion is well-defined.
We present an algorithm for extracting evidences. 

{\bf Evidence extraction:} The evidence $Q$ of $P$ is iteratively computed from the parsing of $P$ as follows.
Let $P=\alpha\beta$ for a maximal prefix $\alpha$ belonging to type1, 2 or 3.
For $i$-th iteration of GC-ESP, $\alpha$ and $\beta$ of $P$ are transformed into $\alpha'$ and $\beta'$, respectively.
In case $\alpha$ is not type2, define $Q_i=\alpha$ and update $P:=\beta'$.
In this case, $Q_i$ is an evidence of $\alpha$ and $\beta'$ is an evidence of $\beta$.
In case $\alpha$ is type2, define $Q_i=\alpha[1,j]$ with $j=\min\{p\mid p\geq \logstar|S|,\; P[p]\mbox{ is landmark}\}$
and update $P:=x\beta'$ where $x$ is the suffix of $\alpha'$ deriving only $\alpha[j+1,|\alpha|]$.
In this case, by Lemma~\ref{lem1:cm}, $Q_i$ is an evidence of $\alpha[1,j]$ and $x\beta'$ is an evidence of $\alpha[j+1,|\alpha|]\beta$.
Repeating this process until $|P|=1$, we obtain the evidence of $P$ as
the concatenation of all $Q_i$.
We obtain the upper bound of length $Q$ as follows.
\begin{lemma}\label{evidence-length}\rm(Maruyama et al.~\cite{ESP})
There is an evidence $Q$ of $P$ such that 
$Q=Q_1\cdots Q_k$ where $Q_i\in q^+_i$ ($q_i\in \Sigma\cup V$, $q_i\neq q_{i+1}$) 
and $k=O(\lg |P|\logstar |S|)$.
\end{lemma}

Thus, we can obtain the time complexity of the pattern finding problem.

{\bf Counting, locating, and extracting:} Given $T_S$ and an evidence $Q$ of $P$, 
a node $z$ in $T_S$ is an occurrence node of $P$ iff there is a sequence $v_1,\ldots, v_k$ such that 
$(v_1,\ldots, v_k)\prec z$ and $L(v_1\cdots v_k)=Q$.
Thus, it is sufficient to adjacently embed all subtrees of $v_1,\ldots, v_k$ into $T_S$.
We recall the fact that the subtree of $v_1$ is left adjacent to that of $v_2$
iff $v_2$ is a leftmost descendant of ${\it right\_child}(lra(v_1))$ where
$lra(v)$ denotes the {\em lowest right ancestor of} $v$, i.e.,
$v$ is the lowest ancestor of $x$ such that the path from $v$ to $x$
contains at least one left edge.
Because $z=lra(v_1)$ is unique and the height of $T_S$ is $O(\lg|S|)$,
we can check whether $(v_1,v_2)\prec z$ in $O(\lg|S|)$ time.
Moreover, $(v_1,v_2,v_3)\prec z'$ iff $(z,v_3)\prec z'$ (possibly $z=z'$).
Therefore, when $|Q_i|=1$ for each $i$, we can execute the embedding of whole $Q$
in $t=O(\lg|P|\lg|S|\logstar|S|)$ time.
For general case of $Q_i\in q^+_i$, the same time complexity $t$ is obtained in Lemma~\ref{evidence-time}.

\begin{lemma}\label{evidence-time}\rm(Maruyama et al.~\cite{ESP})
The time complexity of embedding the evidence of $P$ into $T_S$ is 
$O(\lg|P|\lg|S|\logstar|S|)$.
\end{lemma}

Thus, counting $P$ in $S$ is $O(|P|\lg^*|P| + occ_c\cdot t)$
where $occ_c$ is the frequency of the largest embedded subtree that is called {\em core}.
With a auxiliary data structure storing $|X|$, the length of string derived from $X\in V$,
locating $P$ can be computed in the same time complexity.
Since $T_S$ is balanced, the substring extraction of $S[i,i+m]$ can be computed in $O(m+\lg|S|)$ time.

ESP-index was implemented by LOUDS~\cite{Delpratt06} and permutation~\cite{Munro03} with the time-space trade-off parameter 
$\varepsilon\in (0,1)$, and it supports queries of counting/locating patterns and extracting of substrings.

\begin{theorem}\label{original-ESP-index}\rm(Maruyama et al.~\cite{ESP})
Let $|S|=u$, $|P|=m$, and $n=|V(G)|$ with the GC-ESP $G$ of $S$.
The time for counting and locating is
$O(\frac{1}{\varepsilon}(m\lg n+occ_c\cdot\lg m\lg u)\logstar u)$ 
and the time for extracting substring $S[i,i+m]$ is $O(\frac{1}{\varepsilon}(m+\lg u))$
with $(1+\varepsilon)n\lg n +4n + n\lg u + o(n)$ bits of space and any $\varepsilon\in (0,1)$.
\end{theorem}

\section{ESP-index-I}
We present ESP-index-I for faster query searches than ESP-index. 
ESP-index-I encodes CFGs into a succinct representation by leveraging the idea behind GMR~\cite{Golynski06}, a rank/select dictionary for large alphabets.

{\bf DAG representation:}
we represent a CFG $G$ as a DAG where
$Z\to XY\in P$ is considered as two directed left edge $(Z,X)$ and right edge $(Z,Y)$,
i.e., $G$ can be seen as a DAG with a single source and $|\Sigma|$ sinks.
By introducing a super-sink $s$ and drawing left and right edges from any sink to $s$,
we can obtain the DAG with a single source/sink equivalent to $G$.
We denote the DAG as $DAG(G)$ (Figure~\ref{fig:dag}).
$DAG(G)$ is decomposed into two spanning trees $T_L$ and $T_R$ consisting of 
the left edges and the right edges, respectively.
ESP-index reconstructs $G$ with a permutation $\pi:V(T_L)\to V(T_R)$ from $(T_L,T_R,\pi)$.
Instead, ESP-index-I reconstructs and traverses $G$ by using GMR. 

\begin{figure*}[t]
\begin{center}
\includegraphics[width=0.8\textwidth]{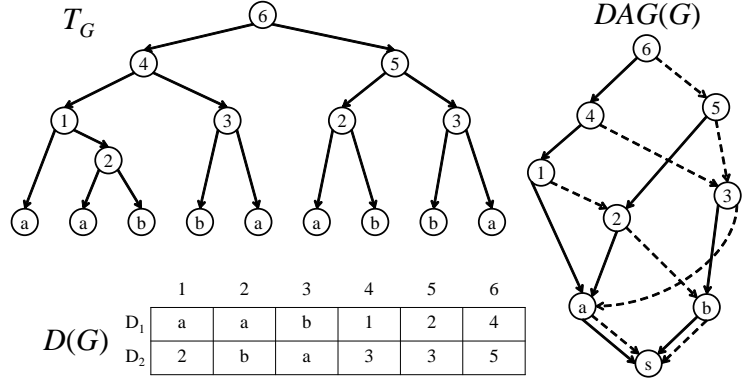}
\end{center}
\vspace{-0.7cm}
\caption{Grammar compression $G$ and its parsing tree $T_G$,
DAG representation $DAG(G)$, and array representation $D(G)$,
where $\Sigma=\{\mbox{a},\mbox{b}\}$ and $V=\{1,2,3,4,5,6\}$.
In $DAG(G)$, the left edges are shown by solid lines.
$D(G)$ itself is an implementation of the phrase dictionary.} 
\label{fig:dag}
\end{figure*}

{\bf Succinct encoding of phrase dictionary:}
For a grammar compression $G$ with $n$ variables,
the set $D(G)$ of production rules is represented by 
a phrase dictionary $D[D_1[1,n],D_2[1,n]]$ 
such that $X_k\to X_iX_j\in D(G)$ iff $D_1[k]=i$ and $D_2[k]=j$.
We consider a permutation $\pi:V\to V$ such that $\pi(D_1)$ is monotonic, i.e.,
$\pi(D_1[i])\leq \pi(D_1[i+1])$.
Then $D$ is transformed into an equivalent $\pi(D)=[\pi(D_1),\pi(D_2)]$
and let $D:=\pi(D)$ (Figure~\ref{fig:dag}).
The monotonic sequence $D_1$ is encoded by the bit vector $B(D_1)$ as follows.

\[
B(D_1) = 0^{D_1[1]}10^{D_1[2]-D_1[1]}1\cdots 0^{D_1[n]-D_1[n-1]}1
\]

By this, we can get $D_1[k]={\rm select}_1(B(D_1),k)-k$ in $O(1)$ time
with $2n+o(n)$ bits of space.
GMR encodes the sequence $D_2$ into $A(D_2)$ with $n\lg n+o(n\lg n)$ bits of space.
We can get $D_2[k]={\rm access }(A(D_2),k)$ in $O(\lg\lg n)$ time.
Thus, we can simulate the phrase dictionary $D$ by $(B(D_1),A(D_2))$.
The access/rank/select on $B(D_1)$ support to traverse $T_L$ and 
the same operations on $A(D_2)$ support to traverse $T_R$.
Thus, we can traverse the whole tree $T_S$ equivalent to $DAG(G)$.

{\bf Simulation of reverse dictionary:}
The improved index, referred as to ESP-index-I, of string $S$ is denoted by $(B(D_1),A(D_2))$.
After indexing $S$, since the memory for $D^{-1}$ is released,
we must construct GC-ESP of pattern $P$ simulating $D^{-1}$ by $(B(D_1),A(D_2))$ 
for counting and locating $P$ in $S$.
To remember $D^{-1}$, the original ESP-index uses the binary search on $T_L$.
On the other hand, we adopt $A(D_2)$ for simulating $D^{-1}$ by an advantage of response time.
Indeed, we can improve the time $O(\lg n)$ to $O(\lg\lg n)$ for a query.
To get $D^{-1}(X_iX_j)=X_k$, we can get the value of $k$ as follows, where let $B=B(D_1)$ and $A=A(D_2)$.

\begin{enumerate}
\item[(1)] Let $p={\rm select}_0(B,i)-i$ and $q={\rm select}_0(B,i+1)-(i+1)$.
\item[(2)] Let $r={\rm select}_j(A,{\rm rank}_j(A,p)+1)$.
\item[(3)] $k=r$ if $r\leq q$, and no $X_k\to X_iX_j$ exists otherwise.
\end{enumerate}

Since $D_1$ is monotonic, we can restrict the range $k\in [p,q]$ by operation (1).
By (2) and (3), we can check if $X_j\in D_2[p,q]$.
If $X_j\in D_2[p,q]$, its position is the required $k$, and
$X_j\not\in D_2[p,q]$, there is no production rule of $X_k\to X_iX_j\in D$.
The execution time of (1), (2), and (3) are $O(1)$, $O(\lg\lg n)$, and $O(1)$, respectively.
The construction of the ESP-index-I is described in Algorithm~\ref{algo:ESP}.

\begin{algorithm}[t]
{\footnotesize
\caption{Construction of ESP-index-I.
$S\in\Sigma^*$: input string, $D=\emptyset$: phrase dictionary.}
\label{algo:ESP}
}
{\footnotesize
\begin{algorithmic}[1]
\Function {ESP-index-I}{}
\While {$|S|>1$}
\State {$D':=$ GC-ESP$(S)$} \Comment{phrase dictionary at each height}
\State {{\rm SORT}$(D')$}   \Comment{renaming for binary search in $D^{-1}$}
\State {$D:=D\cup D'$}
\EndWhile
\State {return $(B(D_1),A(D_2))$}
\EndFunction

\Function{GC-ESP}{S}
\State {set $D=\emptyset$}
\State {execute GC-ESP s.t. $S'(D)=S$}
\State {$S:= S'$}
\State {return $D$}
\EndFunction

\Function{SORT}{D} \Comment{$D=[D_1[1,n],D_2[1,n]]$}
\State {find $\pi:V\to V$ s.t. $\pi(D_1)$ is monotonic} \Comment{$V=\{1,2,\ldots,n\}$}
\State {$D:= [\pi(D_1),\pi(D_2)]$}
\EndFunction

\end{algorithmic}
}
\end{algorithm}

\begin{theorem}\label{esp:th1}
Counting time of ESP-index-I is $O((m+occ_c\lg m\lg u)\lg\lg n\logstar u)$
with $2n+n\lg n + o(n\lg n)$ bits of space.
\end{theorem}
\begin{proof}
By Lemma~\ref{lem2:cm}, the time to build a parsing tree $T_P$ of pattern $P$ is $O(m\lg^* u)$.
The time to simulate the reverse dictionary is $O(\lg\lg n)$ per one input.
Thus, we can find the evidence $Q$ satisfying Lemma~\ref{evidence} in $t_1 = O(m\lg\lg n\lg^* u)$.
By Lemma~\ref{evidence-time}, we can embed $Q$ into $T_S$ in $O(\lg m\lg u\lg^* u)$ time.
To simulate this embedding on $(B(D_1),A(D_2))$, the cost to access the parent-child of any node
is $O(\lg\lg n)$ time.
Thus, the time of the embedding on a maximal core of $occ_c$ times is $t_2 = O(occ_c\lg m\lg u\lg^* u)$.
Then, the counting time is $O(t_1+t_2)$.
The size of the data structures for $B(D_1)$ and $A(D_2)$ are
$2n+o(n)$ and $n\lg n+o(n\lg n)$ bits, respectively.
\end{proof}

\begin{theorem}\label{esp:th2}
With auxiliary $n\lg u + o(n)$ bits of space,
ESP-index-I supports locating in the same time complexity
and also supports extracting in $O((m+\lg u)\lg\lg n)$ time
for any substring of length $m$.
\end{theorem}
\begin{proof}
While ESP-index accesses the parent and children of a node in $1/\varepsilon$ time because it uses a succinct data structure for permutation, 
ESP-index-I takes $O(\lg\lg{n})$ time. 
Thus, the time complexity is obtained from theorem~\ref{original-ESP-index}
by replacing the trade-off parameter $1/\varepsilon$ in ESP-index by $\lg\lg n$.
\end{proof}

\begin{figure}[t]
\begin{center}
\begin{tabular}{cc}
dna.200MB & english.200MB \\
\includegraphics[width=0.48\textwidth]{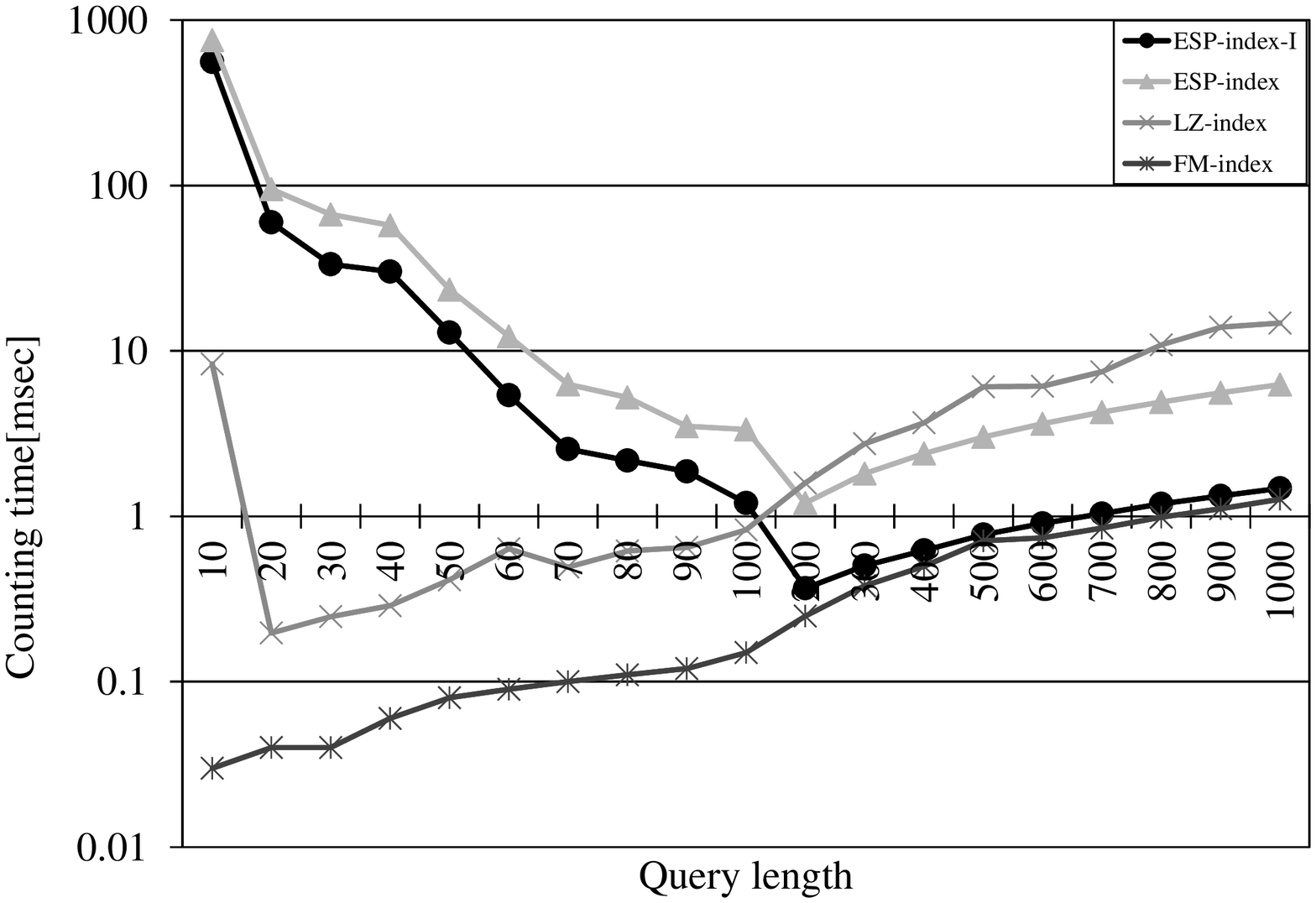} &
\includegraphics[width=0.48\textwidth]{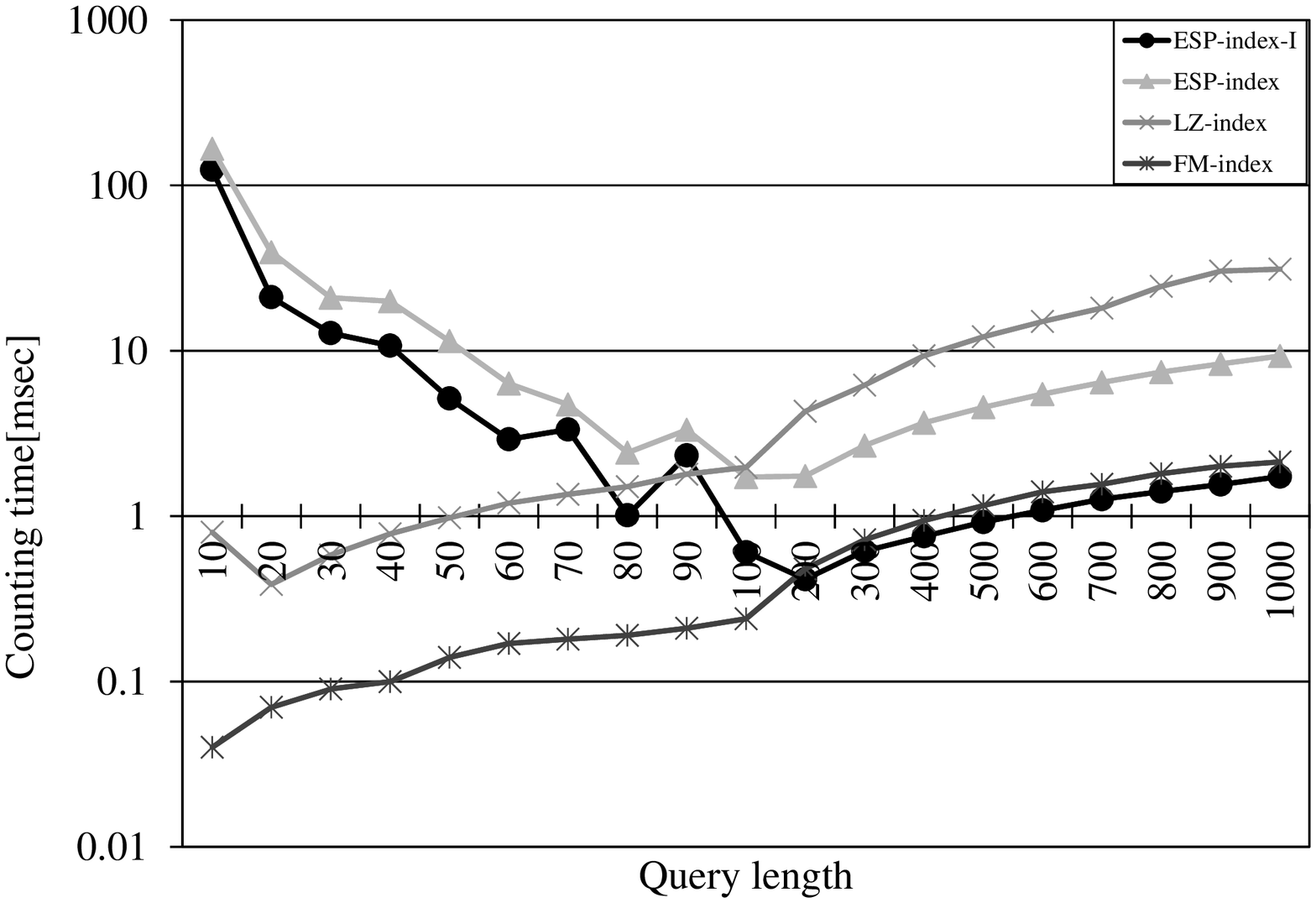} 
\end{tabular}
\end{center}
\vspace{-0.6cm}
\caption{Counting time of each method in milliseconds for dna.200MB (left) and english.200MB (right).}
\label{fig:countpattern}
\begin{center}
\begin{tabular}{cc}
dna.200MB & english.200MB \\ 
\includegraphics[width=0.48\textwidth]{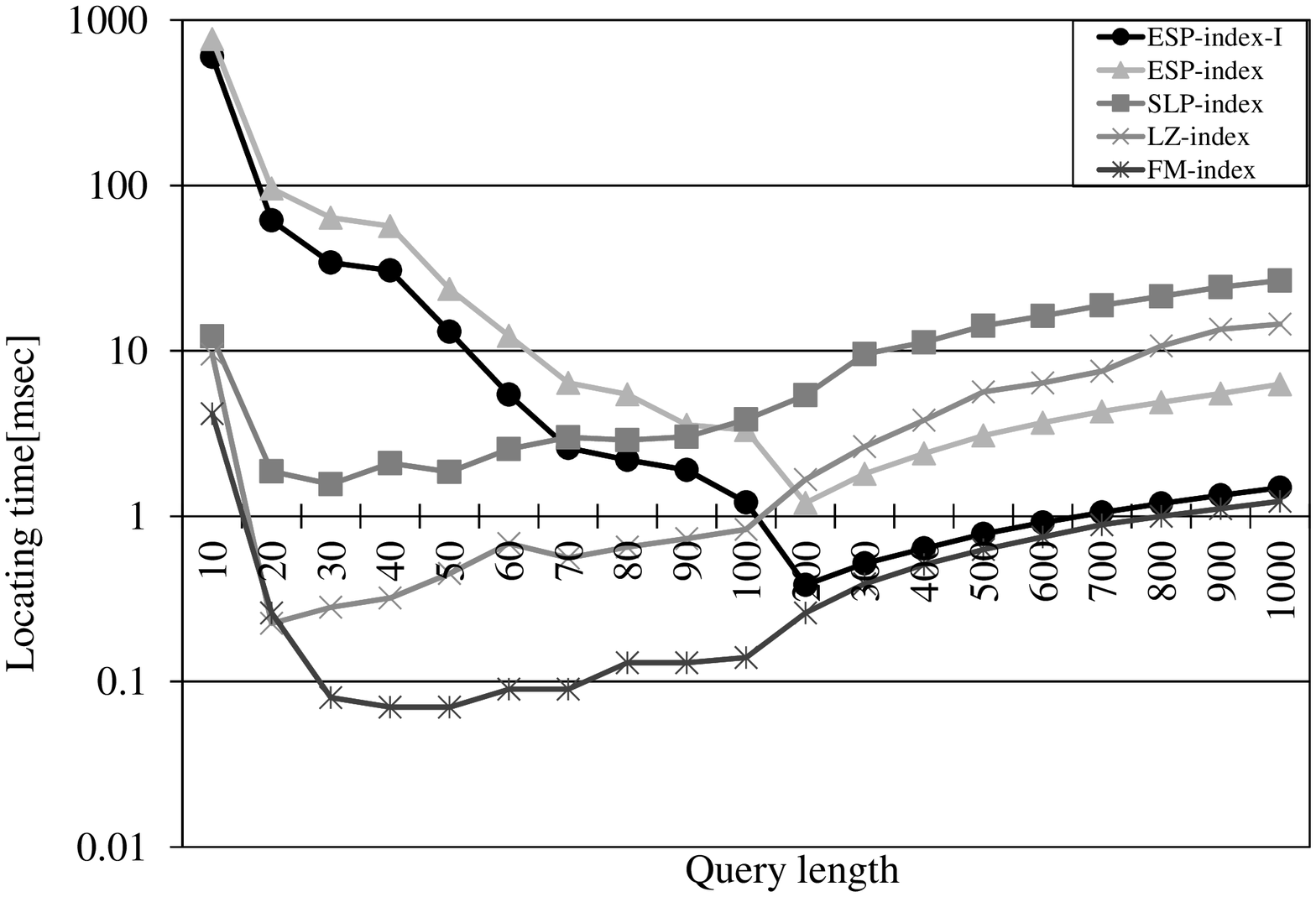} &
\includegraphics[width=0.48\textwidth]{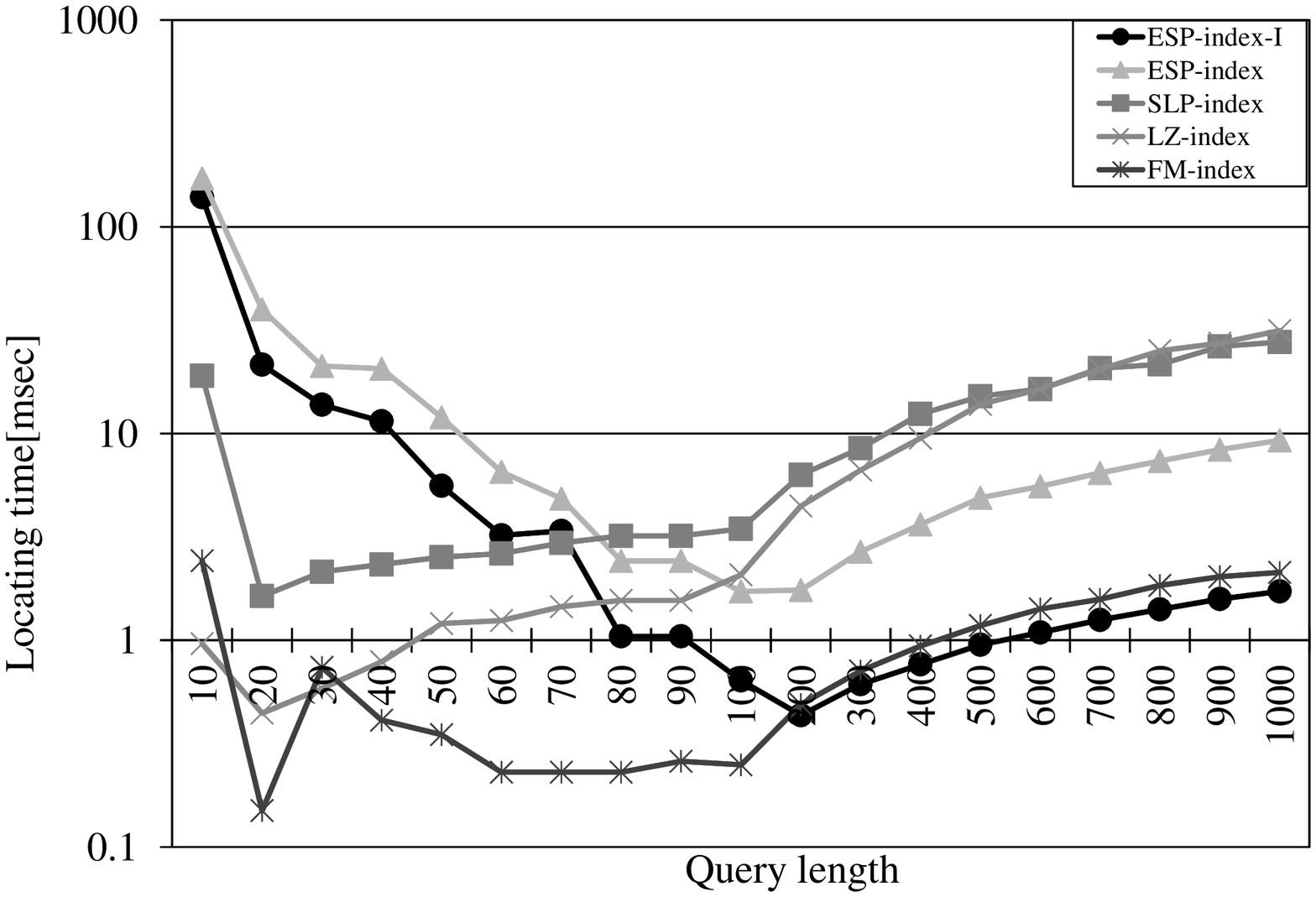} 
\end{tabular}
\end{center}
\vspace{-0.6cm}
\caption{Locating time of each method in milliseconds for dna.200MB (left) and english.200MB (right).}
\label{fig:locatepattern}
\end{figure}

\section{Experiments}
\begin{figure}[t]
\begin{center}
\begin{tabular}{cc}
dna.200MB & english.200MB \\ 
\includegraphics[width=0.48\textwidth]{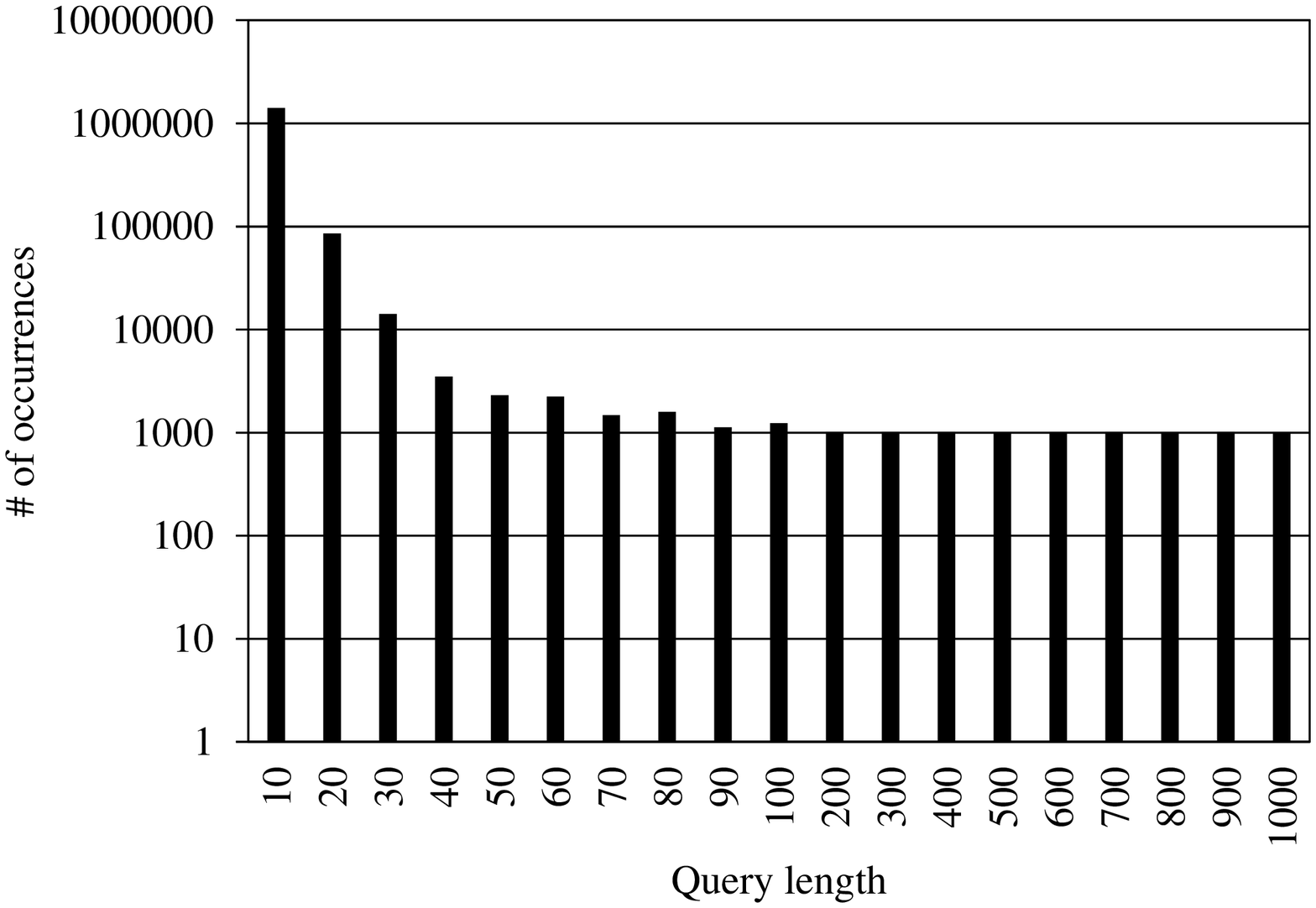} &
\includegraphics[width=0.48\textwidth]{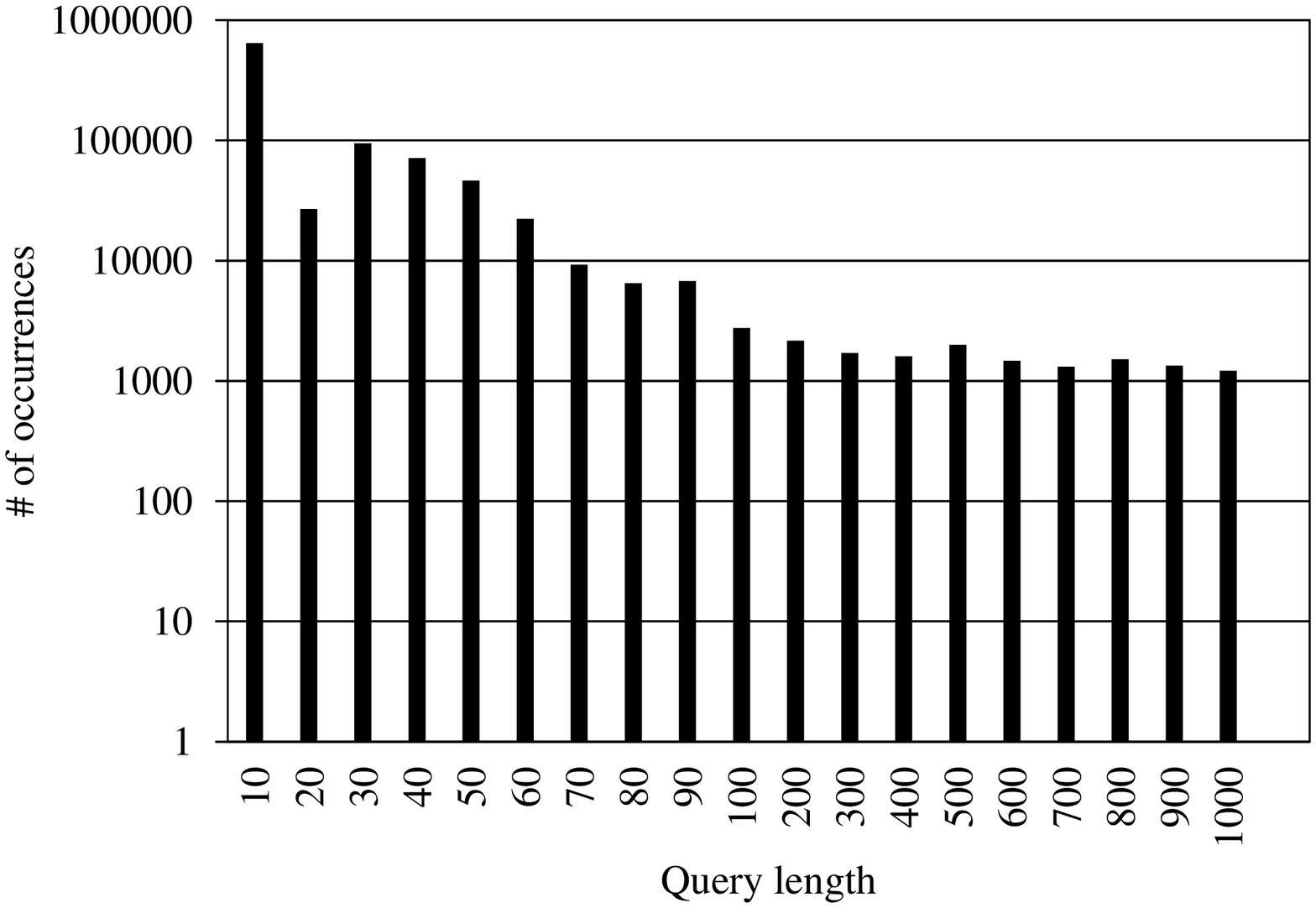} 
\end{tabular}
\end{center}
\vspace{-0.6cm}
\caption{The number of occurrences of query text for dna.200MB (left) and english.200MB (right).}
\label{fig:numocc}
\end{figure}
\begin{figure}[t]
\begin{center}
\begin{tabular}{cc}
dna.200MB & english.200MB \\
\includegraphics[width=0.48\textwidth]{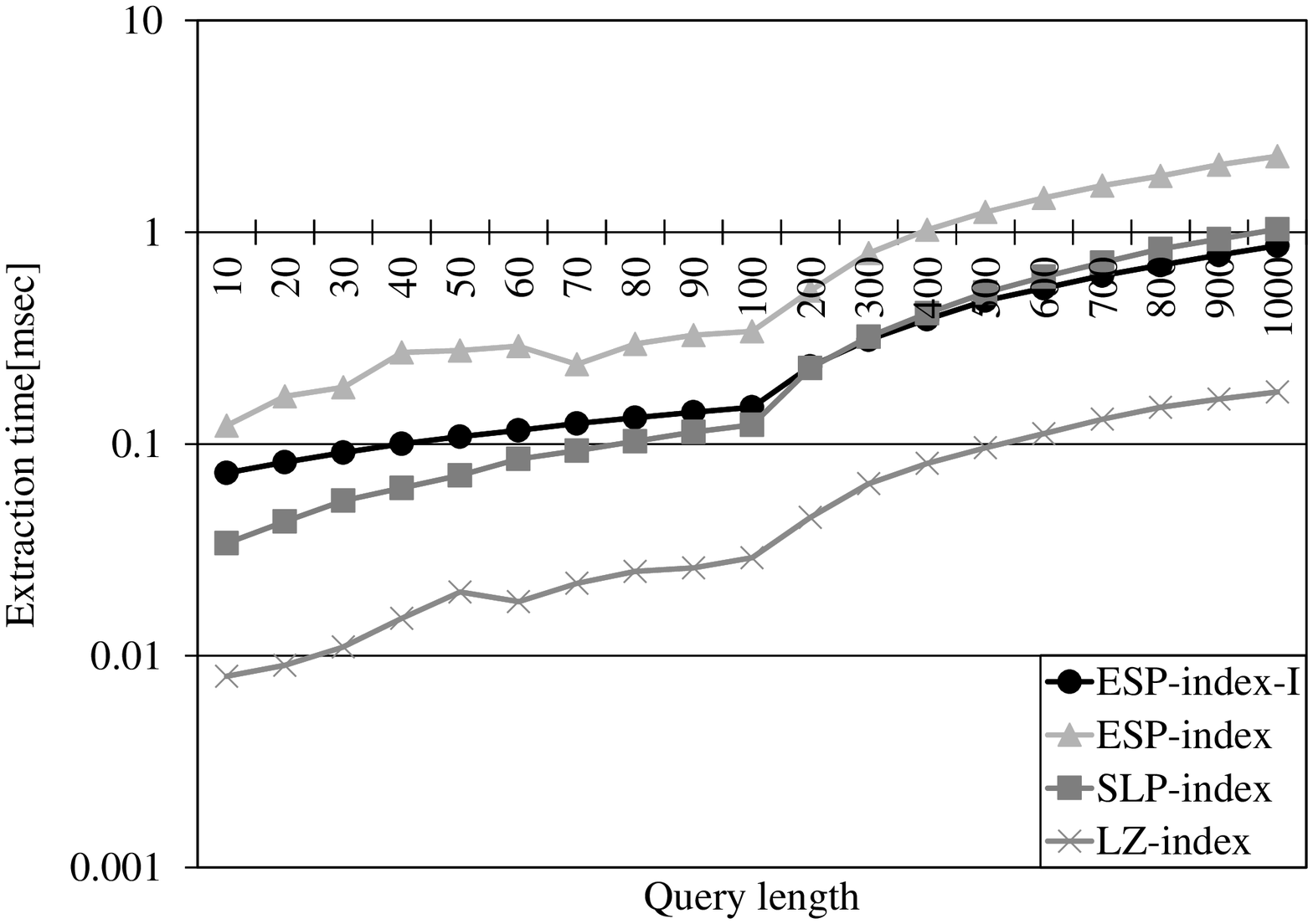} &
\includegraphics[width=0.48\textwidth]{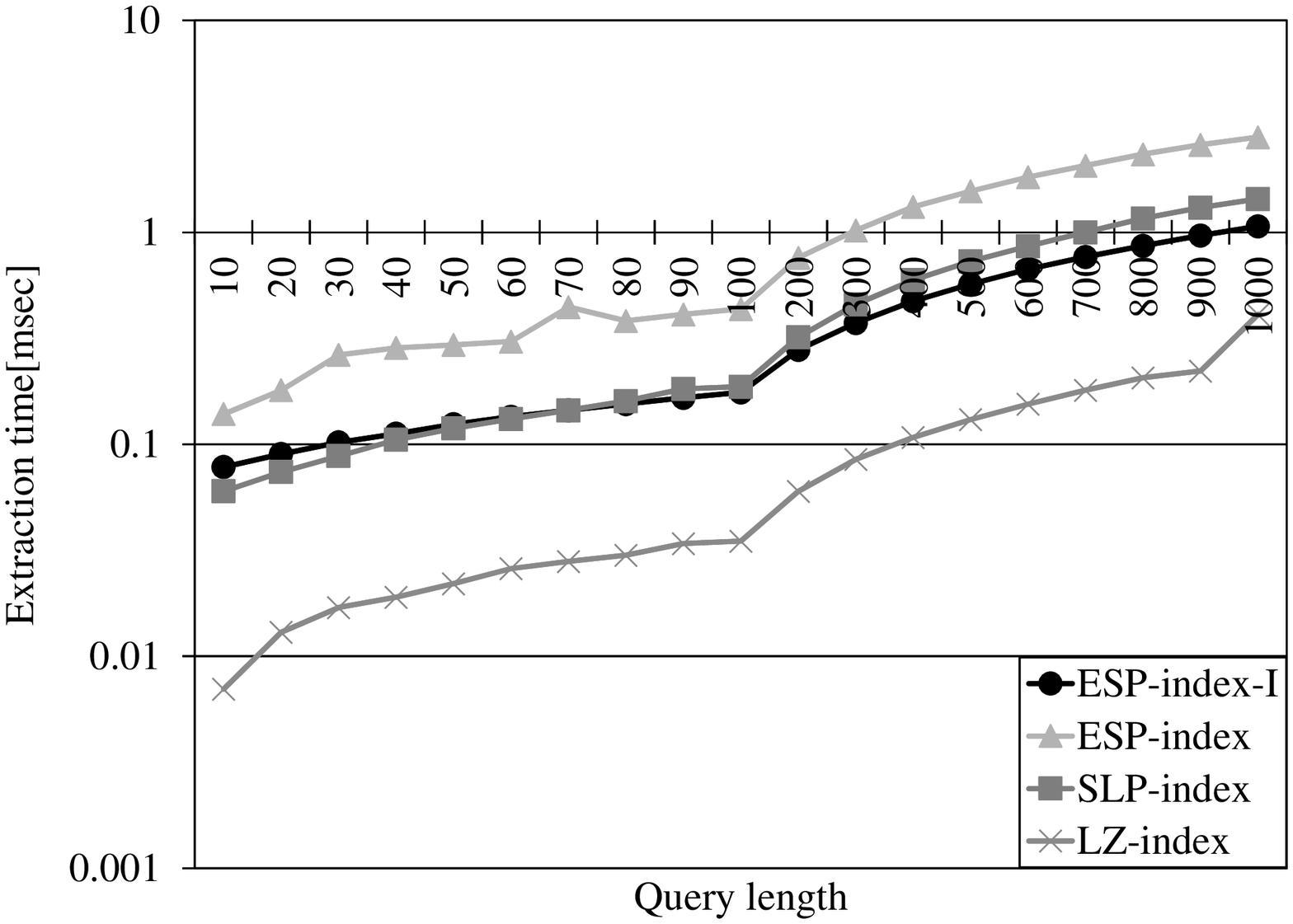} 
\end{tabular}
\end{center}
\vspace{-0.6cm}
\caption{Substring extraction time of each method in milliseconds for dna.200MB (left) and english.200MB (right).}
\label{fig:extract}
\end{figure}

\subsection{Setups}
We evaluated ESP-index-I in comparison to ESP-index, SLP-index, LZ-index on one core of an eight-core Intel Xeon CPU E7-8837 (2.67GHz) machine with 1024 GB memory. 
ESP-index is the previous version of ESP-index-I. 
SLP-index and LZ-index are state-of-the-arts of grammar-based and LZ-based indexes, respectively.
SLP-index used RePair for building SLPs.
We used LZ-index-1 as an implementation of LZ-index.
We also compared ESP-index-I to FM-index, a self-index for general texts, which is downloadable from \url{https://code.google.com/p/fmindex-plus-plus/}.
We implemented ESP-index-I in C++ and the other methods are also implemented in C++. 
We used benchmark texts named dna.200MB of DNA sequences and english.200MB of english texts which
are downloadable from \url{http://pizzachili.dcc.uchile.cl/texts.html}. 
We also used four human genomes\footnote{\url{ftp://ftp.ncbi.nih.gov/genomes/H_sapiens/Assembled_chromosomes/hs_ref_GRC37_chr*.fa.gz}}\footnote{\url{ftp://ftp.ncbi.nih.gov/genomes/H_sapiens/Assembled_chromosomes/hs_alt_HuRef_chr*.fa.gz}}\footnote{\url{ftp://ftp.kobic.kr/pub/KOBIC-KoreanGenome/fasta/chromosome_*.fa.gz}}\footnote{\url{ftp://public.genomics.org.cn/BGI/yanhuang/fa/chr*.fa.gz}} of $12$GB DNA sequences and wikipedia\footnote{http://dumps.wikimedia.org/enwikinews/20131015/enwikinews-20131015-pages-meta-history.xml.bz2} of $7.8$GB XML texts as large-scale repetitive texts.

\subsection{Results on benchmark data}
\begin{table}[t]
\begin{center}
  \caption{Memory in megabytes for dna.200MB and english.200MB.}
  \begin{tabular}{r|c|c|c|c|c}
    & ESP-index-I &  ESP-index & SLP-index & LZ-index & FM-index \\ \hline
dna.200MB      & $156$ & $157$ & $214$ & $208 $ & $325$  \\
english.200MB  & $165$ & $162$ & $209$ & $282$ & $482$ \\
  \end{tabular}
\label{tab:memory}
\end{center}
\end{table}
\begin{table}[t]
\begin{center}
  \caption{Construction time in seconds for dna.200MB and english.200MB.}
  \begin{tabular}{r|c|c|c|c|c}
    & ESP-index-I & ESP-index & SLP-index & LZ-index & FM-index \\ \hline
dna.200MB     &$81.8$  & $82.96$& $1,906.63$ & $64.869$ & $87.7$\\
english.200MB &$93.36$  & $93.58$& $1,906.63$ & $100.624$ & $94.09$\\
  \end{tabular}
\label{tab:consttime}
\end{center}
\end{table}

Figure~\ref{fig:countpattern} and \ref{fig:locatepattern} show the counting and
locating time for query texts consisting of lengths from $10$ to $1,000$ in dna.200MB and english.200MB. 
In addition, the number of occurrences of query texts are presented in Figure~\label{fig:numocc}.
Since the counting and locating time of LZ-index depends quadratically on the query length, counting and locating query texts longer than 200 were slow 
on dna.200MB and english.200MB. 
SLP-index was also slow for locating query texts longer than 200, 
since SLP-index performs as many binary searches as query length for finding the first occurrences of variables. 
ESP-index-I and ESP-index were faster than LZ-index and SLP-index for counting and locating query texts, 
which showed that top-down searches of ESP-index-I and ESP-index for finding occurrences of variables encoding query texts were effective.
ESP-index-I was from $1.4$ to $4.3$ times faster than ESP-index with respect to counting and locating time, which demonstrated our encoding of a parse tree by ESP were effective. 
The counting and locating time of ESP-index-I was compatitive with that of FM-index but ESP-index-I showed a higher compressibility for repetitive texts than FM-index.

Figure~\ref{fig:extract} shows extraction time of subtexts for fixed positions and lengths. 
LZ-index based on LZ77 was fastest among all methods. 
On the otherhand, ESP-index-I was one of the fastest method among grammar-based self-indexes on dna.200MB and english.200MB.

Table~\ref{tab:memory} shows memory usage of each method in megabytes for dna.200MB and english.200MB. 
The methods except FM-index archived small memory, which demonstrated their high compressive abilities for texts. 
The memory usage of ESP-index-I was smallest among all methods, and it used 156MB and 165MB for dna.200MB and english.200MB, respectively. 
Since FM-index is a self-index for general texts, the memory usage of FM-index was largest among methods.
Construction time is presented in Table~\ref{tab:consttime}.

\subsection{Results on large-scale repetitive texts}

\begin{table}[t]
\begin{center}
  \caption{Counting, locating, compression and indexing times in seconds, and index and position size in megabytes for ESP-index-I on large-scale texts.}
  \begin{tabular}{|l|c|c|c|c|}
\hline
    & \multicolumn{2}{c|}{\bf genome} &   \multicolumn{2}{c|}{\bf wikipedia} \\ \hline
      $|P|$         & $200$ & $1,000$ & $200$ & $1,000$\\ \hline
	Counting time~(msec) &$1.06$  & $2.29$  & $139.56$    & $13.04$ \\
	Locating time~(msec) &$1.10$  & $2.33$& $167.40$ & $16.69$ \\
\hline\hline
        Compression time~(sec)     & \multicolumn{2}{c|}{$4,384$} & \multicolumn{2}{c|}{$2,347$} \\
        Indexing time~(sec) & \multicolumn{2}{c|}{$567$} & \multicolumn{2}{c|}{$74$} \\
        Index size~(MB) & \multicolumn{2}{c|}{$3,888$} & \multicolumn{2}{c|}{$594$} \\
        Position size~(MB) & \multicolumn{2}{c|}{$1,526$} & \multicolumn{2}{c|}{$246$} \\
\hline
  \end{tabular}
\label{tab:long}
\end{center}
\end{table}

\begin{figure}[t]
\begin{center}
\begin{tabular}{cc}
genomes & wikipedia \\
\includegraphics[width=0.48\textwidth]{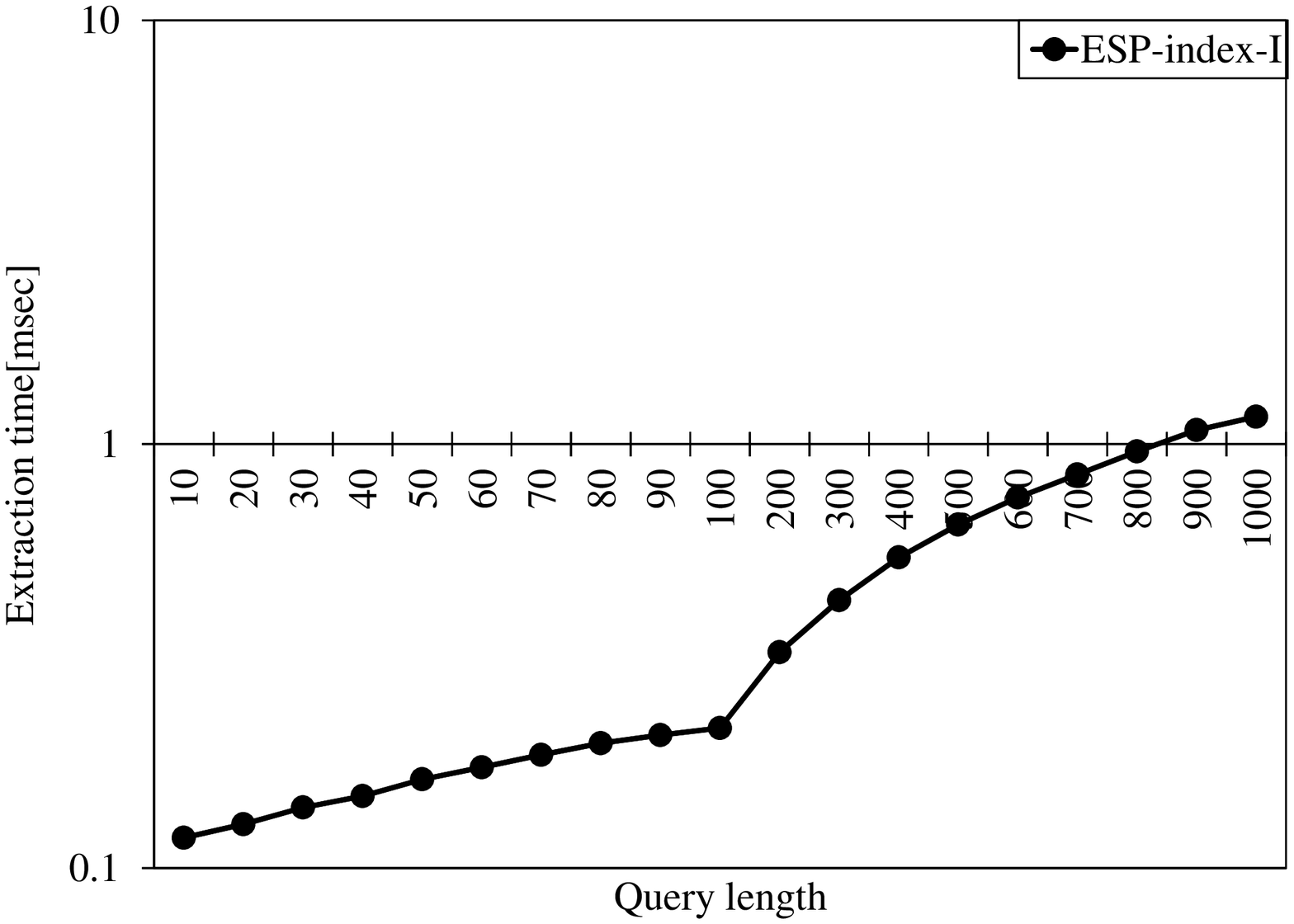} &
\includegraphics[width=0.48\textwidth]{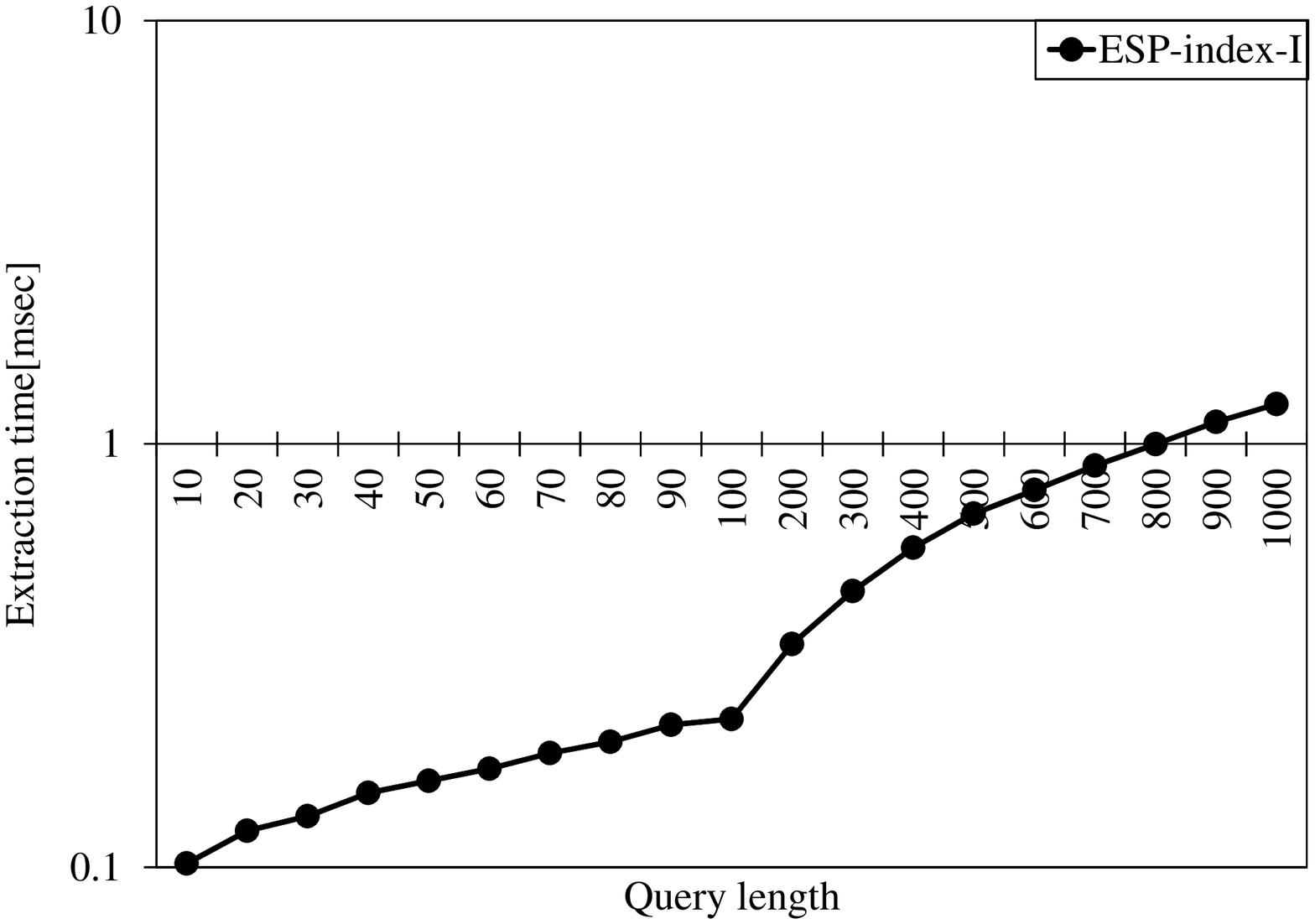} 
\end{tabular}
\end{center}
\vspace{-0.6cm}
\caption{Substring extraction time of four human genomes of $12$GB DNA sequences (left) and wikipedia of $7.8$GB english texts (right).}
\label{fig:longextract}
\end{figure}

We tried ESP-index-I on large-scale repetitive texts. 
The other methods except ESP-index-I did not work on these large texts, because they are applicable to only $32$ bits inputs.
Table~\ref{tab:long} shows the results for counting and locating query texts of lengths of $200$ and $1,000$, construction time and encoding size. 
Figure~\ref{fig:longextract} shows extraction time of substring of lengths from $10$ to $1,000$.
These results showed an applicability of ESP-index-I to real world repetitive texts. 

\section{Conclusion}
We have presented a practical self-index for highly repetitive texts. 
Our method is an improvement of ESP-index
and performs fast query searches by traversing a parse tree encoded by rank/select dictionaries for large alphabets. 
Future work is to develop practical retrieval systems on our self-index. 
This would be beneficial to users for storing and processing large-scale repetitive texts. 

\nocite{Gagie14}

\section{Acknowledgments}
We are thankful to Miguel A. Mart\'inez-Prieto for providing us with a source code of SLP-index. 

{\small
\bibliographystyle{plain}
\bibliography{compress,succinct,cpm,book}
}

\end{document}